\documentclass[10pt,b5paper,twoside, headrule]{amsart}

\usepackage{amsfonts, amsmath, amssymb,latexsym}
\usepackage{epsfig}
\usepackage[curve]{xy}
\usepackage{algorithmic}
\usepackage{algorithm}
\usepackage{enumerate}
\usepackage{framed}
\usepackage{hyperref}
\usepackage{mathtools}
\usepackage[T1]{fontenc}

\usepackage{chngcntr}

\footskip=10mm\parskip 0.2cm\addtocounter{page}{0}
\newtheorem{thm}{Theorem}[section]

\newtheorem{lem}[thm]{Lemma}
\newtheorem{prop}[thm]{Proposition}

\theoremstyle{definition}
\newtheorem{defn}[thm]{Definition}
\newtheorem{examples}[thm]{Examples}

\theoremstyle{remark}
\newtheorem{rem}[thm]{Remark}
\numberwithin{equation}{section}

\newcommand{\Z}{{\mathbb Z}}

\newcommand{\Q}{{\mathbb Q}}

\newcommand{\Id}{\operatorname{Id}}

\newcommand{\Tr}{\operatorname{Tr}}
\newcommand{\cond}{\operatorname{Cond}}

\title[RLWE/ PLWE equivalence]{RLWE/PLWE equivalence for the maximal totally real subextension of the $2^rpq$-th cyclotomic field.}

\author{Iván Blanco Chacón}
\author{Lorena López-Hernanz}
\address{Departamento de Física y Matemáticas, Universidad de Alcalá, Spain}
\email{ivan.blancoc@uah.es; lorena.lopezh@uah.es}
\thanks{First author partially supported by MTM2016-79400-P, CCG20/IA-057, CM/JIN/2019-031 and PID2019-104855RBI00/ AEI/10.13039/501100011033.; second author, by PID2019-105621GB-I00}

\begin{document}
\renewcommand\baselinestretch{1.2}
\renewcommand{\arraystretch}{1}
\def\base{\baselineskip}
\font\tenhtxt=eufm10 scaled \magstep0 \font\tenBbb=msbm10 scaled
\magstep0 \font\tenrm=cmr10 scaled \magstep0 \font\tenbf=cmb10
scaled \magstep0


\def\evenhead{{\protect\centerline{\textsl{\large{I. Blanco}}}\hfill}}

\def\oddhead{{\protect\centerline{\textsl{\large{On the non vanishing of the cyclotomic $p$-adic $L$-functions}}}\hfill}}

\pagestyle{myheadings} \markboth{\evenhead}{\oddhead}

\thispagestyle{empty}

\maketitle

\section{Introduction}
Lattice-based cryptography is one of the most efficient alternatives for the standardisation of postquantum cryptography. Indeed, the majority of surviving proposals in the third round of the NIST public contest belongs to this category. Its strenghts are, first, the ease to implement, and second, the fact that no attack has been found (apart from some weak instances of parameters which can be avoided) which significantly outperforms brute force. Moreover, several theoretical results seem to support a provable hardness guarantee. These results consist on the reduction of several versions of the Shortest Vector Problem for lattices to two of the problems which back lattice-based cryptography: the Learning With Errors Problem (where the reduction comes from the class of general lattices) and the Ring Learning With Errors Problem (where the reduction comes from the subclass of ideal lattices). Despite the fact that the hardness of the precise lattice problems which reduce to the mentioned cryptographic problems has not been established yet,  a promising number of hits has been reached, at least in the category of general lattices (\cite{micciancio}, \cite{khot}). Other feature which makes lattice-based primitives preferable to other approaches is the required size of the secret and public keys to ensure a given security level, far below multivariate-based and code-based contenders. The reader is referred to https://www.safecrypto.eu/pqclounge/ for a detailed description of the remaining proposals and the history of the contest along all the rounds.

Within lattice-based cryptography, the Learning With Errors Problem (LWE from now on), the Ring Learning With Errors Problem (RLWE from now on) and the Polynomial Learning With Errors Problem (PLWE from now on) hold a prominent position. The first was introduced in \cite{regev}, the second in \cite{LPR} and the third in \cite{stehle2}.  In general, PLWE is more suitable for implementations due to the very efficient arithmetic algorithms available for polynomial rings while the majority of security reduction proofs have been established for RLWE. Moreover, at the time of writing and apart from \cite{cls} and those against sheer LWE, there are no direct attacks against RLWE, while a number of theoretical attacks have been produced against PLWE under some general assumptions (see, for instance \cite{elos}, \cite{peikert}). Hence, it seems natural to ask for the relation and mutual dependence of RLWE and PLWE and this is the goal of the present article, which is an extension of \cite{blanco2} and answers a question raised therein by the first author.

In \cite{RSW}, the authors develop the notion of equivalence between RLWE and PLWE in their several versions; namely, both problems are said to be equivalent if there exists an algorithm which transforms admissible RLWE-samples into admissible PLWE-samples and vice versa with a complexity of polynomial order in the rank of the lattice (i.e. the degree of the underlying number field). The word \emph{admissible} means that the algorithm must make the error distribution to incur at most in a distortion which, again, is polynomial in the rank of the lattice.  Moreover, the authors justify why the right measure of this distortion is the condition number of the matrix defining the lattice transformation and study the equivalence for an ad hoc family of polynomials.

The usual (and natural) lattice transformation in \cite{RSW} is nothing else but left multiplication by a Vandarmonde matrix attached to the Galois conjugates of a primitive element of the underlying number field and Vandermonde matrices tend to be very ill-conditioned apart from some ad hoc cases. In the cyclotomic scenario, which is the most dealt with in practical cryptographic primitives, the problem is now reasonably well understood (cf. \cite{DD}, \cite{blanco1}, \cite{italians}).

However, a seek for more flexibility, as justified in \cite{peikert} and in \cite{RSW}, calls for a good understanding of both problems in more general number fields. To our knowledge, the first study of RLWE/PLWE equivalence for a family of non-cyclotomic number fields apart from \cite{RSW} is \cite{blanco2}. In that work, the first author establishes the RLWE/PLWE equivalence for the maximal totally real subextension of the cyclotomic field $\mathbb{Q}(\zeta_{4p})$ (with $p$ arbitrary prime) and justifies why the approach based on the evaluation map at an integral primitive element is deemed to fail due to an exponential lower bound for the condition number of the Vandermonde matrix with real symmetric nodes due to Gautschi (\cite{gautschi}). The main contribution of \cite{blanco2} is the replacement of the Vandermonde matrix by a quasi-Vandermonde matrix attached to the family of Tchebycheff polynomials up to degree $p-1$ and the roots of the $p$-th Tchebycheff polynomials. These matrices are known  to be optimally conditioned and to present a good number of amenable properties due to the orthogonality of the Tchebycheff family.

In the present article we generalize the main result of \cite{blanco2}. Namely, we will prove the following result:
\begin{thm}Let $p$ and $q$ be different odd prime numbers and let $r\geq 2$. For $k=1$, $k=p$ and $k=pq$, the RLWE and the PLWE problems are equivalent for the maximal totally real subextension of the $2^rk$-th cyclotomic field.
\label{thm1our}
\end{thm}
We have structured our presentation in four sections in the following manner:

Section 2 is a summary of algebraic generalities and notations whose aim is to make our article self-contained. Section 3 recalls the RLWE and PLWE problems and the formal definition of \emph{equivalence} and provides a summary of previous results on the equivalence of the R/P-LWE Problems. Subsection 3.4 points out an advantage of using the family $K_{2^rk}^{+}$ (the totally real subextension of the $2^rk$-th cyclotomic field) from a cryptoanalytical point of view: we prove that this family is immune against one of the attacks described in \cite{elos}, whereas cyclotomic fields are not (or at least not in a provable manner).

Section 4 is the core of the paper and proves Theorem \ref{thm1our}. The main ingredient is the analysis of how the condition number behaves under several elementary operations performed in a quasi-Vandermonde matrix attached to the Galois-conjugates of the natural primitive element of the extension.

We thank our colleague Ra\'ul Dur\'an for a careful reading and discussion of our work and for providing us with Example \ref{raul}.

\section{Algebraic setup}

\begin{defn}A lattice is a pair $(\Lambda,\phi)$ where $\Lambda$ is a finitely generated and torsion-free abelian group and $\phi:\Lambda\hookrightarrow\mathbb{R}^N$ is a group monomorphism for some $N$. When $N$ equals the rank of $\Lambda$ we will say that $\Lambda$ has full rank. All our lattices will be full rank unless stated otherwise.
\end{defn}

\subsection{Number fields and ideal lattices}

For any field extension $L/F$, $\mathrm{Gal}(L/F)$ denotes the Galois group of the extension, i.e. the group of field automorphisms of $L$ which fix $F$. 

Let $K = \mathbb{Q}(\theta)$ be an algebraic number field of degree $n$ and let $f(x) \in\mathbb{Q}[x]$ be the minimal polynomial of $\theta$. In particular, $K$ is a $\mathbb{Q}$-vector space of dimension $n$ and the set $\{1,\theta,...,\theta^{n-1}\}$ is a $\mathbb{Q}$-basis. The evaluation-at-$\theta$ map is a field $\mathbb{Q}$-isomorphism $\mathbb{Q}[x]/(f(x))\cong K$. 

The field $K$ is furnished with $n$ field $\mathbb{Q}$-embeddings $\sigma_i: K \hookrightarrow \overline{\mathbb{Q}}$, with $1\leq i\leq n$ and $\overline{\mathbb{Q}}$ a fixed algebraic closure of $\mathbb{Q}$. Each of these morphisms is fully determined by its image at $\theta$, namely $\sigma_i(\theta)=\theta_i$, where $\{\theta_1:=\theta,\theta_2,...,\theta_n\}$ are the roots of $f$ (namely, the Galois conjugates of $\theta$).

The extension $K/\mathbb{Q}$ (or just the field $K$) is said to be Galois if $K$ is the splitting field of $f$. This is equivalent to saying that the embeddings $\sigma_i$ are indeed automorphisms of $K$, hence $\mathrm{Gal}(K/\mathbb{Q})=\{\sigma_1=\Id,\sigma_2,...,\sigma_n\}$.

Setting $s_1$ as the number of real embeddings, i.e. those whose image is contained in $\mathbb{R}$, and $s_2$ as the number of complex non-real embeddings, one has $n=s_1+2s_2$. 
\begin{defn}The canonical embedding $\sigma_K: K\to \mathbb{R}^{s_1}\times\mathbb{C}^{2s_2}$ is defined as:
$$
\sigma_K(x):=(\sigma_1(x),...,\sigma_n(x)).
$$ 
The field $K$ is said to be totally real in case $s_2=0$. When $K$ is clear form the context we will simply write $\sigma$ instead of $\sigma_K$.
\end{defn}

Recall that an algebraic integer is an element of $\overline{\mathbb{Q}}$ whose minimal polynomial belongs to $\mathbb{Z}[x]$. The set $\mathcal{O}_K$ of algebraic integers in $K$ is a ring: the ring of integers of $K$.

It is also well known (see for instance \cite{stewart}) that $\mathcal{O}_K$ is a free $\mathbb{Z}$-module of rank $n$, thus for each ideal $I\subseteq\mathcal{O}_K$ its image $\sigma(I)$ is a lattice in the space
$$
\Lambda_n:=\{(x_1,...,x_n)\in \mathbb{R}^{s_1}\times\mathbb{C}^{2s_2}: x_{s_1+i}=\overline{x}_{s_1+s_2+i}\mbox{ for }1\leq i\leq s_2\}.
$$
Notice that when $K$ is totally real we have $\Lambda_n=\mathbb{R}^n$.

\begin{defn}A lattice $(\Lambda,\phi)$ is said to be an ideal lattice if there exists a number field $K$ and an ideal $I\subseteq\mathcal{O}_K$ such that $\sigma(I)=\phi\left(\Lambda\right)$.
\end{defn}

Of great relevance in cryptography is the obvious observation that every ideal lattice is endowed with an extra ring structure inherited from that in $\mathcal{O}_K$.

\begin{defn}
The field $K$ is said to be monogenic if $\mathcal{O}_K=\mathbb{Z}[\theta]$ for some $\theta\in K$. We will assume that all our fields are monogenic.
\end{defn}

The canonical embedding is one of the two main characters in our story, whose interplay is the object of our study. The second character is presented next:

\begin{defn}The coordinate embedding of $\mathcal{O}_K$ is 
$$
\begin{array}{rcl}
\sigma_{C,K}: \mathcal{O}_K=\mathbb{Z}[\theta] & \to & \mathbb{R}^n\\
a_0+a_1\theta+...+a_{n-1}\theta^n & \mapsto & (a_0,a_1,...,a_{n-1})
\end{array}
$$
\end{defn}
When $K$ is clear from the context we will write $\sigma_C$ instead of $\sigma_{C,K}$. It is worthwhile to mention that multiplication and addition are preserved component-wise by the canonical embedding while, in general, only addition is respected by the coordinate embedding.

\subsection{The cyclotomic field and its maximal totally real subextension} For an integer $n>1$ denote by $\Z_n^*$ the group of multiplicative units in the ring $\Z_n$. The set of
primitive $n$-th roots of unity (those of the form $\zeta_k=\exp(2\pi i k/n)$ with $(k,n)=1$) is a multiplicative group of order $m = \phi(n)$, where $\phi$ stands for Euler’s totient function. The $n$-th cyclotomic polynomial is
$$\Phi_n(x)=\prod_{k\in\Z_n^*}(x-\zeta_k)$$
This polynomial is irreducible in $\mathbb{Z}[x] $ and setting $\zeta=\zeta_k$ for any $k\in\Z_n^*$, the number field
$K_n=\Q(\zeta)$ is the splitting field of $\Phi_n(x)$, hence it is Galois of degree $m$.
In this paper we study the maximal totally real subextension of $K_n$, denoted $K_n^+$, whose degree is $\phi(n)/2$. As $K_n$, the field $K_n^{+}$ is Galois and monogenic (see \cite[Chapter 2]{washington}), namely:
$$
\mathcal{O}_{K_n^+}=\Z[\psi_k],
$$ 
with $\psi_k=\zeta_n^k+\zeta_n^{-k}= 2\cos\left(\frac{2\pi k}n\right)$ for each $k\in\Z_n^*/\{\pm 1\}$. We denote by $\Phi_n^+$ the minimal polynomial of $\psi_k$.

\section{The R/P-LWE problems and the notion of equivalence}


With the same notations as before, denote $\mathcal{O}=\mathbb{Z}[x]/(f(x))$, where $f(x)$ is the minimal polynomial of an integer element $\theta$. As seen in the previous section, the ring $\mathcal{O}$ has a lattice structure in $\mathbb{R}^m$, where $m$ is the degree of $f$, via the coordinate embedding.

\begin{defn}[The search RLWE/PLWE problem] Let $q=r(n)$ be a prime, with $r[x]\in\mathbb{R}[x]$, let $\chi$ be a discrete Gaussian distribution (cf. \cite[Section 2.2]{LPR}) with values in $\mathcal{O}_K/q\mathcal{O}_K$ (resp. in $\mathcal{O}/q\mathcal{O}$). The RLWE (resp. PLWE) problem for $\chi$ is stated as follows:

For a \emph{secret} element $s\in \mathcal{O}_K/q\mathcal{O}_K$ (resp. $\mathcal{O}/q\mathcal{O}$) chosen uniformly at random, if an adversary for whom $s$ is unknown is given access to arbitrarily many samples $\{(a_i,a_is+e_i)\}_{i\geq 1}$ of the RLWE (resp. PLWE) distribution, where for each $i\geq 1$, $a_i$ is uniformly chosen at random and $e_i$ is sampled from $\chi$, this adversary is asked to guess $s$ with non-negligible advantage.
\end{defn}

From now on, by RLWE/PLWE problem we will refer to the search RLWE/PLWE problem. It is a very natural question to wonder what is the relation between the RLWE and the PLWE problems, what we discuss next.

\subsection{The condition number}

In \cite{RSW}, the authors define the notion of equivalence between RLWE and PLWE. In \cite{blanco1} this equivalence is proved for cyclotomic fields under the hypothesis of fixing (or upper bounding) the number of prime divisors of the conductor, and in \cite{blanco2} we have proved the equivalence for the maximal totally real subextension $K_{4p}^{+}$ of the cyclotomic field $K_{4p}$, for $p$ arbitrary prime. 
\begin{defn}For a monogenic Galois number field $K=\mathbb{Q}(\theta)$ of degree $n\geq 2$, the problems RLWE and PLWE are said to be equivalent if each one of them reduces to the other one in polynomial time and with a polynomial noise increase. This means that there exists an algorithm which transfers RLWE-samples into PLWE-samples (and vice versa) with complexity $\mathcal{O}(n^r)$ with $r$ independent of $n$ and this algorithm amplifies the noise by a factor which is also polynomial in $n$. 
\label{equivdef}
\end{defn}

As before, let $f(x)\in\mathbb{Z}[x]$ denote the minimal polynomial of $\theta$ and $\theta_1:=\theta,\theta_2,...,\theta_n$ the Galois conjugates of $\theta$. The evaluation-at-$\theta$ map $V_f$ transforms the lattice $\left(\mathcal{O},\sigma_C\right)$ in the lattice $\left(\mathcal{O}_K,\sigma\right)$:
\begin{equation*}
\begin{array}{ccc}
V_{f}: \mathcal{O} & \to & \sigma_1(\mathcal{O}_K)\times\cdots\times\sigma_n(\mathcal{O}_K)\\
\displaystyle\sum_{i=0}^{n-1}a_i\overline{x}^i & \mapsto & 

\left(\begin{array}{cccc}
1 & \theta_1 & \cdots & \theta_1^{n-1}\\ 
1 & \theta_2 & \cdots & \theta_2^{n-1}\\  
\vdots & \vdots & \ddots \vdots\\ 
1 & \theta_n & \cdots & \theta_n^{n-1}
\end{array}\right)
\left(\begin{array}{c}
a_0\\
a_1\\
\vdots\\
a_{n-1}
\end{array}\right),
\end{array}
\label{latticebij}
\end{equation*}
namely, $V_{f}$ is given by a Vandermonde matrix left-multiplying the vector of coordinates. 

As justified in \cite{RSW} and later in \cite{blanco1}, \cite{blanco2} and \cite{italians}, the idea of \emph{noise increase} can be formally captured by means of the condition number, defined in terms of the Frobenius norm:
\begin{defn}For a square matrix $A=(a_{ij})\in\mathrm{M}_{r}(\mathbb{C})$, the Frobenius norm of $A$  is  
$$
\|A\|:=\sqrt{\Tr(AA^*)}=\sqrt{\sum_{i=1}^r\sum_{j=1}^r|a_{i,j}|^2}.
$$
where $\Tr$ stands for the trace map and $A^*$ is the conjugated-transpose of $A$. The condition number of $A$ is defined as
$$
\mathrm{Cond}(A):=\|A\|\|A^{-1}\|.
$$
\end{defn}

The condition number satisfies the following properties:
\begin{prop}For any $A, B\in \mathrm{GL}_{r}(\mathbb{C})$ it holds:
\begin{itemize}
\item The condition number is invariant by scalar multiplication, namely, for each $\lambda\in\mathbb{C}^*$ it is $\|A\|=|\lambda|\|A\|$ and $\mathrm{Cond}(\lambda A)=\mathrm{Cond}(A)$. 
\item The condition number satisfies $\mathrm{Cond}(A)=\mathrm{Cond}(A^{-1})$.
\item The Frobenius norm and hence the condition number are submultiplicative, namely: 
$$
\mathrm{Cond}(AB)\leq \mathrm{Cond}(A)\mathrm{Cond}(B).$$
\end{itemize}
\label{propcond}
\end{prop}
The condition number captures the idea of \emph{noise increase} caused by the transformation between the lattices $\left(\mathcal{O},\sigma_C\right)$ and $\left(\mathcal{O}_K,\sigma\right)$. Indeed, as proved in \cite{blanco1}, for the cyclotomic field $K_n$, where the transfomation is expressed in terms of the Vandermonde matrix $V_{\Phi_n}$,  the condition number $\cond(V_{\Phi_n})$ is polynomial in $n$ if the number of primes dividing $n$ is fixed.  However, this is not the case for $K_{4p}^{+}$, what led us to replace the map $V_{\Phi_n^{+}}$ by another lattice isomorphism which we proved to be polynomially conditioned in \cite{blanco2}. We will recall this second approach in Section 4.

\subsection{In praise of the family $\Phi_{2^rk}^{+}(x)$}

Despite the fact that both R/PLWE problems are strongly believed to be computationally intractable, several \emph{ad hoc} weak instances have been found and dealt with in a number of recent papers (\cite{elos}, \cite{ehl}, \cite{cls}, \cite{peikert}). As \cite{peikert} points out, these vulnerable instantiations have not been proposed for practical applications, as they do not satisfy the hypotheses of the worst-case hardness theorems which back the R/PLWE cryptosystem proposed in \cite{LPR}. However, quoting \cite{peikert} again, these ad hoc constructions serve to raise the following questions:
\begin{itemize}
\item How \emph{close} are these insecure instantiations to those which enjoy worst-case hardness?
\item Can we identify from these instantiations any feature which make some number fields more secure than others for R/P-LWE?
\item How can we evaluate other instantiations that may not be backed by worst-case hardness theorems?
\end{itemize}

We close this section pointing out a reason to be interested, from a cryptographic point of view, in the family $\Phi_{2^rk}^{+}(x)$, with $r\ge2$ and $k$ odd: we will show that $\Phi_{2^rk}^{+}(x)$ is not vulnerable to one of the attacks described in \cite{elos}, an attack for which cyclotomic polynomials are not immune, in principle. 

The attack has several steps: it starts with a distinguisher attack on PLWE which is transferred to a distinguisher attack against RLWE if RLWE and PLWE are equivalent for the underlying number field. Then, the  decissional RLWE attack is turned into a search attack if two additional hypotheses are satisfied, but we will not enter into it here. 

Soon after \cite{elos}, in \cite{cls} the authors gave a direct attack on RLWE without passing by PLWE, by using the $\chi^2$ statistical test. The attack works for several general cyclotomic fields of non-power-of-two degree, but it does not seem to be any way to apply this attack to our family $\Phi_{2^rm}^{+}(x)$. These attacks also justify the study of other number fields, other than cyclotomics, to instantiate R/PLWE.

The hypotheses for the attack described in \cite{elos} to be effective are as follows:

\begin{thm}[\cite{elos}, \cite{ehl}]\label{th:attack} Let $K=\mathbb{Q}(\beta)$ be a number field of degree $n$ where $\beta\in\mathcal{O}_K$ and let $q$ be an odd prime. Suppose that the pair $(K,q)$ satisfies the following conditions:
\begin{itemize}
\item[1.]$K$ is Galoisian of degree $n$.
\item[2.]The ideal $(q)$ is totally split in $\mathcal{O}_K$.
\item[3.]$K$ is monogenic.
\item[4.]The transformation between the canonical embedding of $K$ and the power basis representation of $K$ is given by a scaled orthogonal matrix.
\item[5.]If $f$ is the minimal polynomial of $\beta$, then:
\begin{itemize}
\item[5.1] either $f(1)\equiv 0 \pmod{q}$,
\item[5.2] or $f(\alpha)\equiv 0\pmod{q}$ for $\alpha\in\mathbb{F}_q$ of \emph{small order} modulo $q$,
\item[5.3] or $f(\alpha)\equiv 0\pmod{q}$ for $\alpha\in\mathbb{F}_q$ of \emph{small residue} modulo $q$.
\end{itemize}
\item[6.] The prime $q$ is large enough, namely, $q>n^2$.
\end{itemize}
Then, there is a polynomial-time attack to the search RLWE problem for $(K,q)$.
\end{thm}

We will not describe the attack here, but we mention that the first two conditions provide the RLWE search-to-decision reduction. The third and fourth conditions are sufficient to grant the RLWE-to-PLWE equivalence, namely, that both problems reduce to each other in polynomial time and with a polynomial error rate distortion (see \cite[Section 4]{RSW}). However, to grant this equivalence, as discussed in \cite{ehl}, it is enough that, when passing from the coordinate to the canonical embedding, the noise increase is polynomial in the degree of the underlying number field, and this noise increase is well accounted for by the condition number of the corresponding matrix. 

Moreover, as we justified in \cite[Remark 2.8]{blanco2}, to grant the RLWE-to-PLWE equivalence it is not necessary to impose that the transformation between both embeddings is the natural one given by the Vandermonde matrix, and for the setting dealt with there, we replaced it by another one, given by a quasi-Vandermonde matrix attached to a subfamily of Tchebycheff polynomials of the first kind.

Finally, the last two conditions are the key to construct the attack on PLWE. Cyclotomic fields are protected against Condition [5.1]: it is well-known that if $n$ is not a prime power then $\Phi_n(1)=1$, and if $n=q^r$ with $q$ prime then $\Phi_n(1)=q$. Therefore, $\Phi(1)\neq 0 \pmod{q}$ unless $n=q^r$ for some $r\geq 1$.

Moreover, cyclotomic fields are also protected against Condition $[5.2]$ for $\alpha\in\mathbb{F}_q$ of order $2$, namely for $\alpha=-1$. Indeed (see for instance \cite[Lemma 7]{cyclotomic}) $\Phi_n(-1)=0$ if $n=2$, $\Phi_n(-1)=p$ if $n=2p^r$, with $p$ prime and $r\ge1$ and $\Phi_n(-1)=1$ otherwise.


It is not clear what can be said, in general, about roots of $\Phi_n(x)$ of order higher than $2$ (but still small), although some partial results have been obtained by the authors in a still ongoing work.

What about Condition $[5.3]$ for the cyclotomic setting? Let $\alpha$ be a root of $\Phi_n(x)$ modulo $q$. Assume that the error distribution is Gaussian, namely, $N(0,\sigma)$ with $\sigma$ chosen in a certain way to grant the ideal-lattice-SVP-to-RLWE reduction and a certain security level (a value of $\sigma\cong 8$ is proposed in \cite{ehl}). Denote by $U$ the event that a sample $(a(x),b(x))$ is taken from the uniform distribution in $\mathcal{O}/q\mathcal{O}\times \mathcal{O}/q\mathcal{O}$ and by $G$ the event that the sample is taken from the PLWE distribution. Denote by $E$ the event that $b(\alpha)-ga(\alpha)$ mod $q$ belongs to the interval $[-q/4,q/4)$ for some guess $g\in \mathbb{F}_q$. In this case (see \cite[p. 10]{ehl}) we have:
$$
p(E|G)=1
$$
and since $p(E|U)=1/2$, then, assuming that samples can be taken from the uniform distribution and from the PLWE distribution with the same probability, it follows that the probability that for some guess $g$ we have $b(\alpha)-ga(\alpha)\in[-q/4,q/4)$ is 3/4. This probability grants an overwhelming probability of success of the attack.

The authors justify that a condition for this to happen is that
\begin{equation}
\frac{\alpha^{2n}-1}{\alpha-1}\leq\frac{q^2}{64\sigma^2},
\label{eqlauter}
\end{equation}
but even if equation~\eqref{eqlauter} does not hold, for several choices of the parameters the attack may work with probability beyond $1/2$:
\begin{examples}[\cite{elos}]
For $n=2^6$, $q\cong 2^{60}$ and $\sigma\cong 8$, and $\alpha=2\pmod{q}$ the authors conclude that their attack works with probability about $0,56$ for any irreducible polynomial (not necessarily cyclotomic) of degree $2^6$ with $\alpha=2$ as a root modulo $q$.
\end{examples}
For $\sigma$ large enough, there is not much hope for equation \eqref{eqlauter} to hold for $\alpha=2$, and even less for $\alpha>2$. For smaller values of $\sigma$, however, the inequality may work for not too large values of $q$. The following examples have been found with the aid of Maple:
\begin{examples}For the cyclotomic polynomial $\Phi_{61}(x)$, $\alpha=2$ is a root modulo $q=2305843009213693951$. For these values, equation \eqref{eqlauter} is satisfied for $\sigma=0.4$. Likewise, for the cyclotomic polynomial $\Phi_{85}(x)$, $\alpha=2$ is a root modulo $q=9520972806333758431$. For these values, \eqref{eqlauter} is again satisfied for $\sigma=0.1$.
\label{raul}
\end{examples}
\begin{rem}
As the first author proved in \cite{blanco1}, denoting by $\omega(n)$ the number of different primes dividing $n$, for every  $M>0$, if $\omega(n)\leq M$ , the condition number of the Vandermonde matrix $V_{\Phi_n}$ is polynomial in $\phi(n)$. Hence, the problems RLWE and PLWE are equivalent for the class of cyclotomic fields $K_n$ if $\omega(n)$ is upper bounded by a fixed value. In particular, both problems are equivalent and hence RLWE is also immune to the attack for $\alpha=\pm 1$ for these fields.
\end{rem}

Now, if we replace the cyclotomic polynomial $\Phi_n(x)$ by the polynomial $\Phi_n^{+}(x)$, at least in the case $n=2^rk$, with $r\ge2$ and $k\ge1$ odd, we can grant not only that $\alpha=\pm 1$ are never roots modulo any odd prime $q$, but also that $\alpha=\pm 2$ is never a root modulo $q$, making the family $\Phi_{2^rk}^{+}(x)$ immune against Condition $[5.1]$, against Condition $[5.2]$ for order $2$, namely for $\alpha=-1$, and against Condition $[5.3]$ for $\alpha=2$. 

Since $4\mid 2^rk$, the polynomial $\Phi_{2^rk}^{+}(x)$ is even (see \cite[Proposition 2.5]{alan}), and hence it is enough to check our claims for $\alpha=1,2$.

\begin{prop} For $\Phi_{2^rk}^{+}(x)$, with $r\ge2$ and $k\ge3$ odd, we have
$$\Phi_{2^r}^+(1)=\pm 1;\quad \Phi_{2^r}^+(2)=2;\quad \Phi_{2^rk}^{+}(1)= \Phi_{2^rk}^{+}(2)= 1.$$
\end{prop}
\begin{proof}

Using \cite[Theorem 2.6]{alan}, we have that $\Phi_{2^r}^+(x)=u_{2^{r-2}}(x)$ and, if $k\ge3$,
$$\Phi_{2^rk}^{+}(x)=\frac{\Phi_k^{+}(u_{2^r}(x))}{\Phi_k^{+}(u_{2^{r-1}}(x))},$$
where $u_n(x):=2t_n\left(x/2\right)$, being $t_n(x)$ the $n$-th Tchebycheff polynomial of the first kind.
Since $u_1(x)=x$, $u_2(x)=x^2-1$ and $u_j(u_l(x))=u_{jl}(x)$ for all $j,l\ge0$, we obtain that $u_{2^n}(1)=-1$ for all $n\ge1$ and $u_{1}(1)=1$, so $\Phi_{2^r}^+(1)=1$ if $r=2$, $\Phi_{2^r}^+(1)=-1$ if $r\ge3$ and
$\Phi_{2^rk}^{+}(1)= 1$. To see that $\Phi_{2^r}^+(2)=2$ and $\Phi_{2^rk}^{+}(2)= 1$ we can use the same argument, taking into account that $u_n(2)=2$ for all $n\geq 1$ (see \cite[Corollary 2.4]{alan}). We can also give the following alternative proof: for each $n\geq 1$, consider the rational expression
$$
r_n(x):=\Phi_n^{+}(x+x^{-1})x^{\frac{\phi(n)}{2}}.
$$
Since $x+x^{-1}=x^{-1}(x^2+1)$, we see that $r_n(x)$ is a polynomial with integer coefficients. Moreover, since $\Phi_n^{+}(\zeta_n+\zeta_n^{-1})=0$, then $r_n(x)$ vanishes at $\zeta_n$ and hence $\Phi_n(x)\mid r_n(x)$. Since the degree of $r_n(x)$ is precisely $\phi(n)$, then $r_n(x)=\Phi_n(x)$ up to a non-zero rational scalar. Moreover, is is easy to see that this scalar factor is the leading coefficient of $\Phi_{n}^{+}(x)$, which is $1$. Hence
$\Phi_{2^rk}^{+}(2)=\Phi_{2^rk}(1)$, which is equal to 2 if $k=1$ and equal to 1 if $k\ge3$.
\end{proof}

Hence, if for $K_{2^rk}^{+}$ both RLWE and PLWE are equivalent, then also the RLWE problem will be immune against this attack based on the roots $\alpha=\pm 1, \pm 2$ for any odd prime $q$, while $K_{2^rk}$ is only provably immune against the attack based on $\alpha=\pm 1$. The goal of the rest of our article is to prove the equivalence between both problems in the sense of Definition \ref{equivdef} for $k=1$, $k=p$ and $k=pq$, with $p$ and $q$ arbitrary odd primes and $p\neq q$.

\section{The polynomial equivalence in the $2^rpq$ maximal totally real cyclotomic subextension}

Here we establish the equivalence of the RLWE and PLWE problems for the family $\Phi_{2^rk}^{+}(x)$ in the generalised framework described in \cite{blanco2} for $k=1$, $k=p$ and $k=pq$. The starting point of our approach there and also here is the family of Tchebycheff polynomials of the first kind:

\begin{defn}The family of Tchebycheff polynomials of the first kind is defined by any of the following equivalent properties:
\begin{itemize}
\item[a)] $t_i(x)=\cos(i\arccos(x))$ for $i\geq 0$.
\item[b)] $t_0(x)=1, t_1(x)=x$ and $t_i(x)=2xt_{i-1}(x)-t_{i-2}(x)$ for $i\geq 2$.
\end{itemize}
\label{tcheby}
\end{defn}
Set, for $i\ge0$, $u_i(x)=2t_i(x/2)$. An easy induction argument shows that $u_i(x)\in\mathbb{Z}[x]$ for each $i\geq 0$. By an iterated application of the identity for the cosine of the sum of two angles, it is easy to show the following property:
$$
t_{i+j}(x)+t_{|i-j|}(x)=2t_i(x)t_j(x),
$$ 
which yields the following identity:
\begin{equation}\label{eq:recursion}
u_{i+j}(x)+u_{|i-j|}(x)=u_i(x)u_j(x)
\end{equation}
for any $i,j\ge0$.

For $n\geq 1$ such that $4\mid n$, set $N=n/4$ and denote 
$$V_N=\left(u_i(\psi_{2k-1})\right)_{0\le i\le N-1\atop 1\le k\le N},$$
where $\psi_{2k-1}=2\cos\left(\frac{2\pi(2k-1)}{n}\right)$. Using \cite[Theorem 1]{kuian} it was shown in \cite[Proposition 3.4]{blanco2} that 
\begin{equation}\label{eq:boundVN}
\cond(V_{N})\leq N(N+1).
\end{equation}

\subsection{Main result}
\subsubsection{Case $n=2^r$} Assume that $n=2^r$, with $r\geq 2$ so that $N=2^{r-2}=\deg(K_{2^r}^{+})$. In this case, $\Phi_n^{+}(x)$ is the minimal polynomial of $\psi_1=2\cos\left(\frac{\pi}{2^{r-1}}\right)$, whose Galois conjugates are all the nodes $\psi_{2k-1}$ with $1\leq k\leq 2^{r-2}$. Hence the matrix $V_N$ already provides the polynomially conditioned lattice monomorphism from $(\mathcal{O},\sigma_C)$ to $(\mathcal{O}_{K_n^{+}},\sigma)$. Since the image of $V_N$ is a finite order sublattice of $(\mathcal{O}_{K_n^{+}},\sigma)$, multiplying by this index and composing with $V_N^{-1}$ provides a polynomially conditioned monomorphism from $(\mathcal{O}_{K_n^{+}},\sigma)$ to $(\mathcal{O},\sigma_C)$. This will be discussed in Theorem \ref{mainthm}, which applies to the three cases under study, namely $n=2^r$, $n=2^rp$ and $n=2^rpq$.

\subsubsection{Case $n=2^rp$} Assume that $n=2^rp$, with $p$ an odd prime number, so $N=2^{r-2}p$ and $m=\deg(K_{2^rp}^{+})=2^{r-2}(p-1)$. Since we are looking for a polynomially conditioned lattice isomorphism between $(\mathcal{O},\sigma_C)$ and $(\mathcal{O}_{K_n^{+}},\sigma)$, we need to exclude from $V_N$ the $N-m=2^{r-2}$ rows corresponding to the values $\psi_{2k-1}$ such that $(2k-1,p)\neq1$, i.e. the nodes of the form
$$\psi_{jp}=2\cos\left(\frac{j\pi}{2^{r-1}}\right),$$
and we also need to extract $2^{r-2}$ columns so that the remaining $m\times m$ matrix is invertible. This is equivalent to multiplying $V_N$ by a suitable matrix $A$ obtaining
$$
V_NA=\left(\begin{array}{cc}
V_{N-m}  & O\\
C & R_m
\end{array}
\right),
$$
where $O$ is the $(N-m)\times m$ zero matrix,  $V_{N-m}$ is a square matrix of dimension $N-m$, $C$ is an $m\times (N-m)$ matrix and $R_m$ is a square invertible matrix of size $m$ with entries in $\mathcal{O}_{K_n^{+}}$. If $A$ is polynomially conditioned, as we will prove, so will be $R_m$.

We begin with a permutation of the rows of $V_N$ so that those $2^{r-2}$ rows corresponding to the nodes $\psi_{jp}$ are in the first positions. For simplicity, we still denote this matrix by $V_N$, which clearly has the same condition number as before.

The strategy, as in \cite{blanco2}, will be to perform certain transformations on the columns of $V_N$ such that the $(N-m)\times m$ upper-right submatrix of $V_N$ becomes the zero matrix and to control the condition number of the matrix which contains all the elementary column operations. Write $V_N=[C_0, C_1,\dots, C_{N-1}]$, where $C_i$ is the $(i+1)$-th column of $V_N$. Observe that $u_{2^{r-2}}(\psi_{jp})=0$ for all $j$, so the first $N-m$ entries of the column $C_{2^{r-2}}$ are zero. We perform the following transformations:
for each $l$ with $1\le l\le 2^{r-2}(p-1)-1$, we replace the column $C_{2^{r-2}+l}$ by the sum
\begin{equation*}\label{eq:columns1}
C_{2^{r-2}+l}+C_{|2^{r-2}-l|}
\end{equation*}
This transformation is equivalent to multiplying $V_N$ on the right by the matrix
$$
A=\left(
\begin{array}{cc}
I_{N-m} & K\\
O & L_m
\end{array}
\right),
$$
where $I_{N-m}$ is the identity matrix of dimension $N-m$, $K$ is a $(N-m)\times m$ matrix whose columns are zero except for the $k$-th ones, with $2\le k\le 2^{r-1}$, all of which have one entry equal to 1 and the rest of them equal to 0, and $L_m$ is matrix whose $k$-diagonals are zero except for $k=0$ and $k=2^{r-1}$ (where by $k$-diagonal of a square matrix $(a_{ij})$ we mean the diagonal whose terms are the entries of the form $a_{i, i+k}$), both of which are constant and equal to 1.  

By relation~\eqref{eq:recursion}, after this transformation the term corresponding to the node $\psi_{2k-1}$ in the new column $C_{2^{r-2}+l}$ is the value at $\psi_{2k-1}$ of $u_{2^{r-2}}u_l$, so it is zero on $\psi_{jp}$. Regarding the condition number of $A$, we have that

\begin{prop}\label{pro:condition1}
$$\cond (A)<\sqrt{3m(5m^2+2m)}.$$
\end{prop}
\begin{proof}
Since that matrix $A$ has only $N+2^{r-1}-1+m-2^{r-1}=N+m-1$ non-zero terms, all of them equal to 1, and $N<2m$, we have
$$\|A\|^2<3m.$$
The inverse of $A$ is given by
$$A^{-1}=\left(
\begin{array}{cc}
I_{N-m} & -KL_m^{-1}\\
O & L_m^{-1}
\end{array}
\right).$$
The matrix $L_m^{-1}$ has $k$-diagonal equal to zero except if $k=j2^{r-1}$, for $j\ge0$, in which case the diagonal is constant and equal to $(-1)^j$. Therefore, the entries of $KL_m^{-1}$ are all bounded in absolute value by 2, and
$$\|A^{-1}\|^2<N+4(N-m)m+\frac{m(m-1)}2<5m^2+2m,$$
which shows the bound for the condition number of $A$.
\end{proof}

\subsubsection{Case $n=2^rpq$}\label{subsub:2^rpq} Assume that $n=2^rpq$, with $p, q$ odd prime numbers and $p<q$, so $N=2^{r-2}pq$ and $m=\deg(K_{2^rpq}^{+})=2^{r-2}(p-1)(q-1)$. In this case, we need to exclude from $V_N$ the $N-m=2^{r-2}(p+q-1)$ rows corresponding to the values $\psi_{2k-1}$ such that $(2k-1,pq)\neq1$, i.e. the nodes of the form
$$\psi_{jp}=2\cos\left(\frac{j\pi}{2^{r-1}q}\right) \text{ and } \psi_{jq}=2\cos\left(\frac{j\pi}{2^{r-1}p}\right),$$
with $j$ odd such that $jp\le 2N-1$ and $jq\le 2N-1$ and we also need to extract $N-m$ columns so that the remaining $m\times m$ matrix is invertible. 

As in the previous case, we begin with a permutation of the rows of $V_N$ so that those $N-m$ rows corresponding to the nodes $\psi_{jp}$ and $\psi_{jq}$ are in the first positions, and we still denote this matrix by $V_N$.

We will use the following fact:
\begin{lem}\label{eq:zerosum}
Notations as before, we have that $u_{2^{r-2}q}(\psi_{jp})=0$ for all $j$, and 
$$\sum_{k=1}^{\frac{p-1}{2}}(-1)^{k+1}u_{2^{r-1}k}(\psi_{jq})=1$$
for all $j$ such that $p\nmid j$.
\end{lem}
\begin{proof}
The first identity is clear. For the second one, we use Lagrange trigonometric identity:
$$
\sum_{k=1}^M\cos\left(k\theta\right)=-\frac{1}{2}+\frac{\sin((M+1/2)\theta)}{2\sin(\theta/2)},\mbox{ for any }M\geq 1\mbox{ and any }0<\theta<2\pi.
$$
Notice that the identity remains true for any $\theta$ which is not an integer multiple of $2\pi$. Hence, taking $M=(p-1)/2$ and $\theta=(j+p)\pi/p$ (which is not a multiple of $2\pi$ since $p\nmid j$), we have:
$$
\sum_{k=1}^{\frac{p-1}2}(-1)^{k+1}\cos\left(\frac{jk\pi}{p}\right)=-\sum_{k=1}^{\frac{p-1}2}\cos\left(\frac{jk\pi}p+k\pi\right)=
\frac{1}{2}-\frac{\sin((j+p)\pi/2)}{2\sin((j+p)\pi/(2p))}=\frac{1}{2}$$
since $j$ and $p$ are odd. Therefore,
\begin{align*}
\sum_{k=1}^{\frac{p-1}{2}}(-1)^{k+1}u_{2^{r-1}k}(\psi_{jq})&=2\sum_{k=1}^{\frac{p-1}{2}}(-1)^{k+1}\cos\left(\frac{kj\pi}p\right)=1.\qedhere
\end{align*}
\end{proof}

Write $V_N=[C_0, C_1,\dots, C_{N-1}]$, where $C_i$ is the $(i+1)$-th column of $V_N$. We perform the following transformations:
for each $l$ with $0\le l\le m-1$, we replace the column $C_{2^{r-2}(q+p-1)+l}$ by the sum
\begin{equation}\label{eq:firstcolumn}
\sum_{j=1}^{\frac{p-1}2}(-1)^{j+1}\left[C_{2^{r-2}(q+p-(2j-1))}+C_{2^{r-2}(q-p+(2j-1))}\right]+(-1)^{\frac{p+3}2}C_{2^{r-2}q}\quad \mbox{ if }l=0,
\end{equation}
and by the sum
\begin{align}\label{eq:nextcolumns}
\sum_{j=1}^{\frac{p-1}2}&(-1)^{j+1}\left[C_{2^{r-2}(q+p-(2j-1))+l}+C_{|2^{r-2}(q-p+(2j-1))-l|}+C_{2^{r-2}(q-p+(2j-1))+l}+C_{|2^{r-2}(q+p-(2j-1))-l|}\right]\\
+\,&(-1)^{\frac{p+3}2}\left[C_{2^{r-2}q+l}+C_{|2^{r-2}q-l|}\right]\nonumber\quad\mbox{ if }l\geq 1. 
\end{align}
This transformation is equivalent to multiplying $V_N$ on the right by the matrix
$$
A=\left(
\begin{array}{cc}
I_{N-m} & K\\
O & L_m
\end{array}
\right),
$$
where $I_{N-m}$ is the identity matrix, $K$ is a $(N-m)\times m$ matrix and $L_m$ is an upper triangular matrix whose main diagonal is the identity.

By relation \eqref{eq:recursion}, after this transformation the term corresponding to the node $\psi_{2k-1}$ in the new column $C_{2^{r-2}(q+p-1)+l}$ is the value at $\psi_{2k-1}$ of 
$$u_{2^{r-2}q}\left[\sum_{j=1}^{\frac{p-1}2}(-1)^{j+1}u_{2^{r-1}\frac{p-(2j-1)}2}+(-1)^{\frac{p+3}2}\right]=(-1)^{\frac{p+1}2}u_{2^{r-2}q}\left[\sum_{j=1}^{\frac{p-1}2}(-1)^{j+1}u_{2^{r-1}j}-1\right]$$
in case $l=0$, or
\begin{align*}
&u_{2^{r-2}q}\left[\sum_{j=1}^{\frac{p-1}2}(-1)^{j+1}\left(u_{2^{r-2}(p-(2j-1))+l}+u_{2^{r-2}(p-(2j-1))-l}\right)+(-1)^{\frac{p+3}2}u_l\right]\\
=\,&u_{2^{r-2}q}u_l\left[\sum_{j=1}^{\frac{p-1}2}(-1)^{j+1}u_{2^{r-1}\frac{p-(2j-1)}2}+(-1)^{\frac{p+3}2}\right]\\
=\,&(-1)^{\frac{p+1}2}u_{2^{r-2}q}u_l\left[\sum_{j=1}^{\frac{p-1}2}(-1)^{j+1}u_{2^{r-1}j}-1\right],
\end{align*}
in case $l\ge1$, so by Lemma~\ref{eq:zerosum} it is zero on $\psi_{jp}$ and $\psi_{jq}$. 

We have the following result for the blocks $K$ and $L_m$ of the matrix $A$:
\begin{lem}\label{lem:matrixA}
The matrix $K$ has entries in $\{0,\pm 1,\pm 2\}$, and the matrix $L_m$ is a Toeplitz matrix whose $k$-diagonal is equal to $(-1)^j$ if $k=j2^{r-1}$ for $0\le j\le p-1$, is equal to $(-1)^{j+1}$ if $k=j2^{r-1}$ with $q\le j\le q+p-1$ and is equal to 0 otherwise.
\end{lem}
\begin{proof}
The condition on the entries of $K$ is clear for the first column, whose non-zero entries are in fact $\pm 1$, since the columns that appear in the sum \eqref{eq:firstcolumn} are all different. For the other columns of $K$, observe that for all $1\le i,j\le (p-1)/2$ we have
\begin{equation}\label{eq:columninequalities}
2^{r-2}(q+p-(2i-1))+l> 2^{r-2}q+l> 2^{r-2}(q-p+(2j-1))+l>l.
\end{equation}
Moreover, for all $1\le j\le (p-1)/2$ we have that 
\begin{enumerate}[-]
\item If $l\le 2^{r-2}q$ then $|2^{r-2}(q+p-(2j-1))-l|>|2^{r-2}q-l|$ and either
$$|2^{r-2}(q-p+(2j-1))-l|<l \quad\text{or}\quad |2^{r-2}(q-p+(2j-1))-l|<|2^{r-2}q-l|;$$
\item If $l>2^{r-2}q$ then $l>|2^{r-2}(q-p+(2j-1))-l|>|2^{r-2}q-l|$.
\end{enumerate}
These relations, together with \eqref{eq:columninequalities}, imply that each column in the sum \eqref{eq:nextcolumns} appears at most twice, and the statement on the entries of $K$ is proved. 

To show the statement for $L_m$, note that for each $l\ge 1$ an index $k\ge 2^{r-2}(q+p-1)$ in the sum \eqref{eq:nextcolumns} can only be of the form $|2^{r-2}(q-p+(2j-1))-l|$ if $l>2^{r-2}(q-p+(2j-1))$ (otherwise those terms are less than $2^{r-2}(q-2)$), can only be of the form $ |2^{r-2}(q+p-(2j-1))-l|$ if $l> 2^{r-2}(q+p-(2j-1))$ (otherwise those terms are less than $2^{r-2}(q+2)<2^{r-2}(q+p-1)$) and can only be of the form $|2^{r-2}q-l|$ if $l>2^{r-2}q$ (otherwise those terms are less that $2^{r-2}q$). In this case, for all $1\le i,j\le (p-1)/2$ we have
$$l-2^{r-2}(q-p+(2i-1))>l-2^{r-2}q>l-2^{r-2}(q+p-(2j-1)).$$
This relation, together with \eqref{eq:columninequalities}, shows that each column $C_k$ with $k\ge 2^{r-2}(q+p-1)$ appears at most once in the sum \eqref{eq:nextcolumns}. Observe that the maximum index that appears in the sum \eqref{eq:nextcolumns} among the terms of the form $2^{r-2}(q+p-(2j-1))+l$, $2^{r-2}q+l$ and $2^{r-2}(q-p+(2j-1))+l$ is $2^{r-2}(q+p-1)+l$ and the minimum is $2^{r-2}(q-p+1)+l$, so those terms range the $k$-diagonals for $0\le k\le 2^{r-1}(p-1)$; the difference between two consecutive terms, which appear with opposite sign in the sum \eqref{eq:nextcolumns}, is $2^{r-1}$, so the statement for the diagonals of the form $2^{r-1}j$ with $0\le j\le p-1$ follows. Analogously, the maximum index that appears in the sum \eqref{eq:nextcolumns} among the terms of the form  $l-2^{r-2}(q-p+(2j-1))$, $l-2^{r-2}q$ and $l-2^{r-2}(q+p-(2j-1))$ is $l-2^{r-2}(q-p+1)$ and the minimum is $l-2^{r-2}(q+p-1)$, so those terms  range the $k$-diagonals for $2^{r-1}q\le k\le 2^{r-1}(q+p-1)$, and the statement for the columns of the form $2^{r-1}j$ with $q\le j\le q+p-1$ follows as in the previous case.
\end{proof}
As for the condition number, we have:
\begin{prop}\label{cond11}
With notations as before,
$$
\cond(A)<\sqrt{(5m^2+2m)(4m^4+m^2+2m)}.
$$
\end{prop}
\begin{proof} By Lemma~\ref{lem:matrixA}, since the non-zero entries of $K$ are bounded in absolute value by 2 and the non-zero entries in $L_m$ are bounded in absolute value by 1, we have
$$\|A\|^2\leq N+4(N-m)m+\frac{m(m-1)}2.$$
Since
$$\frac N m=\frac{p}{p-1}\frac{q}{q-1}<2$$
we obtain $\|A\|^2< 5m^2+2m$.
On the other hand, 
$$A^{-1}=\left(
\begin{array}{cc}
I_{N-m} & -KL_m^{-1}\\
O & L_m^{-1}
\end{array}
\right).$$
Since $L_m$ is an upper triangular Toeplitz matrix, if we denote $r(x)=1+\sum_{j=1}^{m-1} a_jx^j$, where $a_j$ is the entry of the $j$-diagonal of $L_m$, the inverse of $L_m$, which is also an upper triangular Toeplitz matrix, has $k$-diagonal equal to $b_k$, where $s(x)=\sum_{k\ge0}b_kx^k$ is the formal series such that $r(x)s(x)=1$. Denote $t=2^{r-1}$.
If $p=3$, then $(p-1)t\le m-1< qt$, so by Lemma~\ref{lem:matrixA} we have that
$r(x)=1-x^t+x^{2t}$. Since $r(x)(1+x^t)=1+x^{3t}$, then 
$$s(x)=(1+x^t)(1+x^{3t})^{-1}=(1+x^t)\sum_{i\ge0}(-1)^ix^{3it}$$
and the non-zero diagonals of $L_m^{-1}$ are $\pm1$. If $p\ge5$ then $(q+p-1)t\le m-1$, since
$$\dfrac{(q+p-1)t}m=\dfrac{2(q+p-1)}{(p-1)(q-1)}\le\frac{2(q+4)}{4(q-1)}<1.$$ 
Then, by Lemma~\ref{lem:matrixA} we have that
$$r(x)=1-x^t+x^{2t}-x^{3t}+\dots+x^{(p-1)t}+x^{qt}-x^{(q+1)t}+x^{(q+2)t}-x^{(q+3)t}+\dots+x^{(q+p-1)t}.$$
Since $r(x)(1+x^t)=1+x^{pt}+x^{qt}+x^{(p+q)t}=(1+x^{pt})(1+x^{qt})$, then
$$s(x)=(1+x^t)(1+x^{pt})^{-1}(1+x^{qt})^{-1}=\left(1+x^{t}\right)\sum_{i,j\ge0}(-1)^{i+j}x^{(ip+jq)t}.$$
Note that if $ip+jq=kp+lq+1$ for some $i,j,k,l\ge0$, then $i+j$ and $k+l$ have opposite parity, so the coefficient of $x^{(ip+jq)t}$ in the series $s(x)$ is zero; moreover, $ip+jq=kp+lq$ with $(i,j)\neq(k,l)$ can only hold if $ip+jq\ge pq$, so $(ip+jq)t\ge m$. This implies that the non-zero diagonals of $L_m^{-1}$ can only be $\pm 1$. Therefore the entries of $-KL_m^{-1}$ are bounded in absolute value by $2m$. Then,
$$\|A^{-1}\|^2\le N+4m^2(N-m)m+\frac{m(m-1)}2<4m^4+m^2+2m$$
and the bound for the condition number follows.
\end{proof}

\subsubsection{Conclusion}

Putting together the cases $n=2^r$, $n=2^rq$ and $n=2^rpq$, due to Proposition \ref{propcond}, from equation~\eqref{eq:boundVN} and Propositions~\ref{pro:condition1} and \ref{cond11} we obtain
\begin{align*}\label{cond12}
\cond(V_N)&\leq m(m+1) \quad \text{ if } n=2^r \\[2pt]
\cond(V_NA)&\leq 2m(2m+1)\sqrt{3m(5m^2+2m)}\quad \text{ if } n=2^rp \\[2pt]
\cond(V_NA)&\leq 2m(2m+1)\sqrt{(5m^2+2m)(4m^4+m^2+2m)}\quad \text{ if } n=2^rpq.
\end{align*}
Let us write
$$
V_NA=\left(
\begin{array}{cc}
V_{N-m} & O\\
C & R_m
\end{array}
\right),
$$
where $V_{N-m}$ is the principal minor of $V_N$ of order $N-m$, which is invertible, $C$ is an $m\times(N-m)$ matrix and $R_m\in \mathrm{M}_{m\times(N-m)}(\mathcal{O}_{K_{2^rpq}^{+}})$ is invertible (since $V_N$ is so). The reason why the entries of $R_m$ belong to $\mathcal{O}_{K_{2^rpq}^{+}}$ is that they are linear combinations of the entries of $V_N$ with coefficients $0,\pm 1,\pm 2$. Hence, it is also immediate that
\begin{equation*}
\|R_m\|\leq \|V_NA\|.
\label{cond13}
\end{equation*}
Now, the inverse of $V_NA$ exists and is:
$$
(V_NA)^{-1}=\left(
\begin{array}{cc}
V_{N-m}^{-1} & O\\
-R_m^{-1}CV_{N-m}^{-1} & R_m^{-1}
\end{array}
\right),
$$
hence
\begin{equation*}
\|R_m^{-1}\|\leq \|(V_NA)^{-1}\|.
\label{cond14}
\end{equation*}
The previous inequalities imply that
\begin{equation*}\label{cond14}
\cond(R_m)\leq \cond(V_NA).
\end{equation*}
We are now in position to prove our main result:
\begin{thm}Let $p$ and $q$ be different odd prime numbers and let $r\geq 2$. For $n=2^r$, $n=2^rp$ and $n=2^rpq$, the RLWE and the PLWE problems are equivalent for the ring of integers $\mathcal{O}_{K_n^{+}}$ of $K_n^{+}$.
\label{mainthm}
\end{thm}
\begin{proof}Setting as before $n=2^rpq$ (the other cases are analogous), the map
$$
\begin{array}{rcl}
\Psi: \mathcal{O} & \to & \sigma(\mathcal{O}_{K_{2^rpq}^{+}})\\
\textbf{u}& \mapsto & R_m\textbf{u}
\end{array}
$$
is a well defined monomorphism of lattices, since $R_m\in\mathrm{M}_{m\times m}(\mathcal{O}_{K_{2^rpq}^{+}})$ and it is invertible. Moreover, as established above, the condition number of $R_m$ is $O(20m^5)$. This provides a polynomial reduction from PLWE to RLWE incurring into a distortion which is polynomial in $m$.

The map $\Psi$ is not necessarily surjective since we can only grant that $\Psi(\mathcal{O})$ is a sublattice of $\sigma(\mathcal{O}_{K_{2^rpq}^{+}})$ but since $\Psi$ is injective (because $R_m$ is invertible), the ranks of $\Psi(\mathcal{O})$ and $\sigma(\mathcal{O}_{K_{2^rpq}^{+}})$ coincide. In particular, the image $\Psi(\mathcal{O})$ is a sublattice of $\sigma(\mathcal{O}_{K_{2^rpq}^{+}})$ of finite index.

Indeed, as it was shown in section~\ref{subsub:2^rpq}, there exist polynomials $p_j(x)\in\mathbb{Z}[x]$, for $0\le j\le m-1$, of different degrees such that each row of the matrix $R_m$ is of the form $(p_0(\psi_l), p_1(\psi_l)), \dots, p_{m-1}(\psi_l))$ for some odd $l$ coprime with $p$ and $q$. 
Moreover, the polynomials $p_j$ are of the form
$$p_0(x)=a(x); \quad p_j(x)=a(x)u_j(x) \text{ for } 1\le j\le m-1,$$
where 
$$a(x)=(-1)^{\frac{p+1}2}u_{2^{r-2}q}(x)\left[\sum_{i=1}^{\frac{p-1}2}(-1)^{i+1}u_{2^{r-1}i}(x)-1\right].$$
Note that $a(x)$ is a polynomial of degree $N-m$ and its roots are exactly the $N-m$ nodes $\psi_{jp}$ and $\psi_{jq}$, with $j$ odd, so $a(\psi_l)\neq0$ for every odd $l$ coprime with $p$ and $q$. Consequently, the elements $\{p_0(\psi_1), p_2(\psi_1), \dots, p_{m-1}(\psi_1)\}$ are linearly independent over $\mathbb{Z}$, since otherwise the elements $\{1, u_1(\psi_1), \dots, u_{m-1}(\psi_1)\}$ would be linearly dependent, which is clearly not the case, as $\deg(u_i(x))=i$ and $\deg(K_{2^rpq}^+)=m$.

Hence, for $0\leq i\leq m-1$, taking into account that the minimal polynomial of $\psi_1$ has degree $m$, we can write   $p_i(\psi_1):=\sum_{j=0}^{m-1}a_{i,j}\psi_1^{j}$ so that the matrix
$$
P=(a_{ij})_{0\le i,j\le m-1}
$$
contains the coordinates (with respect to the power basis of $\mathcal{O}_{K_{2^rpq}^{+}}$) of the elements $\{p_0(\psi_1), p_1(\psi_1), \dots, p_{m-1}(\psi_1)\}$, which are a $\mathbb{Z}$-basis of a sublattice $\Lambda$ of $\mathcal{O}_{K_{2^rpq}^{+}}$ of the same rank and hence of finite index $\lambda:=|\mathrm{det}(P)|$ (see \cite[Theorem 1.17]{stewart}).

Now, pushing forward this inclusion by the canonical embedding $\sigma$, which is in particular a lattice isomorphism over its image, we observe that $\Psi(\mathcal{O})=\sigma(\Lambda)$ is a sublattice of $\sigma(\mathcal{O}_{K_{2^rpq}^{+}})$ of the same index $\lambda$.

Hence, for each $\textbf{u}\in \sigma(\mathcal{O}_{K_{2^rpq}^{+}})$, one has that $\lambda \textbf{u}\in \Psi(\mathcal{O})$ and the map
$$
\begin{array}{rcl}
\Psi^{-1}\circ [\lambda]: \sigma(\mathcal{O}_{K_{2^rpq}^{+}}) & \to & \mathcal{O}\\
\textbf{u}& \mapsto & R_m^{-1}(\lambda \textbf{u})
\end{array}
$$
provides a a polynomial reduction from RLWE to PLWE incurring into a distortion which is, since homotheties have condition number 1, polynomial in $m$.
\end{proof}

\end{document}